\documentclass[epj]{svjour}
\usepackage{graphics}
\usepackage{amsmath}

\usepackage[usenames,dvipsnames]{xcolor} % extra colors

\begin{document}
\title{Given enough choice, simple local rules percolate discontinuously}
\author{Alex Waagen\inst{1} \and Raissa M. D'Souza\inst{1,2}% etc
%\thanks{\emph{Present address:} Insert the address here if needed}%
}                     % Do not remove
\institute{University of California, Davis, CA 95616, USA \and The Santa Fe Institute, Santa Fe, NM 87501, USA}
\date{Received: date / Revised version: date}
% The correct dates will be entered by Springer
%
\abstract{
%Percolation phase transitions have been shown to have rich variety universality classes and w
%
% with distinct nature and mechanisms
%
%and their underlying mechanisms 
There is still much to discover about the mechanisms and nature of
discontinuous percolation transitions. 
Much of the past work considers
graph evolution algorithms known as Achlioptas processes in which
a single edge is added to the graph from a set of $k$ randomly chosen
candidate edges at each timestep until a giant component emerges. Several
%unbounded 
Achlioptas processes seem to yield a discontinuous percolation
transition, but it was proven by Riordan and Warnke that the transition
must be continuous in the thermodynamic limit. However, they also
proved that if the number $k(n)$ of candidate edges increases with
the number of nodes, then the percolation transition may be discontinuous.
Here we attempt to find the simplest such process which yields a discontinuous
transition in the thermodynamic limit. We introduce a process which
considers only the degree of candidate edges and not component size.
We calculate the critical point $t_{c}=(1-\theta(\frac{1}{k}))n$
and rigorously show that the critical window is of size $O\left(\frac{n}{k(n)}\right)$.
If $k(n)$ grows very slowly, for example $k(n)=\log n$, the critical
window is barely sublinear and hence the phase transition is discontinuous
but appears continuous in finite systems. We also present arguments that Achlioptas processes with 
%that even with infinite choice, 
bounded size rules will always have
continuous percolation transitions even with infinite choice. 
 % end of PACS codes
} %end of abstract
\maketitle
\section{Introduction}
\label{intro}

Percolation describes the onset of large scale connectivity amongst
sites on a lattice or nodes in a network, and 
%percolation theory is
mathematical models of percolation serve as 
an underpinning for analyzing properties of networks, including epidemic
thresholds, vulnerability, and robustness~\cite{stauffer1985introduction,NW1999,Cohen2000,Calloway2000,MoorePRE,Newman:2010:NI:1809753}.
In a prototypical process, one starts from a collection of $n$ isolated
nodes, and then edges are chosen uniformly at random from the set of all possible edges and sequentially added to the graph. 
A set of nodes connected together by following a path along edges is considered a {\it component}, so initially all components are of size one.  
As the number of edges increases, approaching $\frac{n}{2}$, a giant component (i.e., a component linear in system size $n$)
emerges in a continuous phase transition.  As the behaviors of a system can be radically different if large-scale connectivity
%, rather than limited connectivity, 
exists,  altering the location and nature of the percolation phase transition has been an outstanding challenge.

% RD:  Such details of other models are distracting here since the point of discussing them is not mentioned.   Are you trying to make the point that a rich range of behaviors can be found? 
%
% There are several examples of percolation processes
%in which the giant component emerges discontinuously~\cite{Janson_phasetransitions,chen2011exp}.
%Most recently it was shown that one such process yields a continuous
%transition in the initial emergence of a giant component, but with
%a discontinuous jump in the size of this giant component arbitrarily
%close to the transition point. As opposed to the typical power law
%scaling, the giant component grows as a devil's staircase~\cite{PhysRevX.2.031009}.

Around the year 2000, the notion of delaying or enhancing the onset of percolation by a small variant on the standard process was introduced in a procedure which is now referred to as an ``Achlioptas process''~\cite{riordanAnnals2012}. %,Bohman01avoidinga}.
Rather than sampling a single edge at a time, \textit{two}
edges are sampled simultaneously but only the edge that best satisfies
a pre-specified criteria is added to the graph and the other edge discarded. %and in this context it is often referred to as an Achlioptas process.
The criteria used to select the winning edge typically considers the
sizes of the components that would be joined by the edge. For example,
to enhance the onset of percolation, choose the edge that maximizes
the size of the resulting component. To delay, choose the edge that
minimizes it. 
This is an example of the concept of the ``power of two choices'' which has previously yielded considerable benefits
in randomized algorithms \cite{azar1999,Mitzenmacher00thepower,Mitzenmacher:2001:PTC:504336.504343,Sinclair:2010:DSR:1886521.1886576}.

Initial analysis of Achlioptas processes established their effectiveness in delaying percolation~\cite{Bohman2001}.  But only more recently was it shown that the nature of the phase transition can be altered by such processes~\cite{Achlioptas13032009}.  In~\cite{Achlioptas13032009} they analyze the ``Product Rule'', an Achlioptas process where two edges are examined simultaneously, but only the edge that minimizes the product of the sizes of the two components  to be joined by the edge is added to the graph. All initial numerical evidence of the Product Rule indicated that the phase transition to large scale connectivity is discontinuous.
%, yet successive works established that the maximum change in connectivity from the addition of a single edge, though large, decreases with increasing $n$~\cite{nagler11,manna2011new}. 
Yet, analytic arguments~\cite{PhysRevLett.105.255701,nagler11,PhysRevLett.106.225701,manna2011new,lee2011continuity,tian2012nature} and a rigorous proof~\cite{Riordan15072011} now show that, although the transition appears discontinuous for any finite system,  the transition is in fact continuous in the thermodynamic limit of system size $n\rightarrow \infty$.  However the transition belongs to a universality class distinct from standard continuous percolation transitions~\cite{PhysRevLett.106.225701,tian2012nature}. 

Beyond the Product Rule, it is now known that any Achlioptas process which chooses between two edges, or even chooses amongst any fixed number of edges, yields a continuous percolation transition in the thermodynamic limit~\cite{Riordan15072011}.  On the other hand, several related models have now been shown to have truly discontinuous percolation transitions~\cite{chen2011exp,panagiotou2011explosive} or more exotic behaviors such as an initially continuous transition followed by a discontinuous jump in the size of the giant component arbitrarily close to the first transition point~\cite{PhysRevX.2.031009}. And moreover, even the first transition point may be discontinuous if the process is restricted to a lattice or some other structure. It was recently shown that a discontinuous transition may be obtained in Euclidean space with an Achioptas rule which is defined with respect to spanning clusters. However, in order for the transition to be discontinuous with a fixed number of choice $m$, the dimension $d$ must be less than $6$. If $d > 6$, then the number of choices  $m$ must be at least on the order of $\log n$~\cite{Cho2013science}.    

Given the broad array of work on this topic and the lack of complete
understanding, it is important to isolate the essential ingredients
necessary for a discontinuous percolation transition. In \cite{Riordan15072011}
they prove that if instead of a constant number of choices, $k$,
the number is a function $k(n)\rightarrow\infty$ as $n\rightarrow\infty$,
the percolation transition can be discontinuous. However, this discontinuous
transition is considered trivial~\cite{PhysRevX.2.031009} because
it occurs 
%RD:  Rewrote this because $t$ had not been defined.
%at time 
when the number of edges added (denoted by $t$) satisfies the condition $t=n$. In most processes, there is a high probability of merging distinct components at each individual step, and the number of distinct components at time t will be very nearly $n - t + 1$. Thus, a transition occurring at $t = n$ most often indicates that the only reason a giant component emerged at all was because the process ``ran out" of components to merge. 

Here we show that if $k(n)\rightarrow\infty$
as $n\rightarrow\infty$ then %much simpler 
an extremely simple rule, which considers only node
\textit{degree}, is sufficient to yield a truly discontinuous transition, 
which happens at critical point $t_{c}=(1-\theta(\frac{1}{k}))n$.
Furthermore, the simplicity of the rule allows us to rigorously bound
the scaling window which we show has width $O(\frac{n}{k})$. Although
the transition can be shown rigorously to be discontinuous in the
thermodynamic limit, for any finite size $n$, if $k(n)$ increases
extremely slowly with $n$, (e.g., $k(n)=\log n$), then the transition will not have a single large
discrete jump. % as seen with the Product Rule. 
This is the opposite of what has been observed with the Product Rule which exhibits large discrete jumps in any finite system, thus appearing to be discontinuous,  
%which yields a transition
%that appear discontinuous 
but actually being continuous in the thermodynamic limit.

The remainder of the manuscript is organized as follows. In Sec.~\ref{sec:1} we formally present the Achlioptas processes studied herein, in particular distinguishing between unbounded size rules, such as the Product Rule, and bounded size rules which treat all components of size greater than a specified value $M$ as equivalent. 

Sec.~\ref{sec:results} contains the bulk of the results discussed above in addition to 
%Furthermore,  we present 
arguments showing that bounded size rules
will always have continuous percolation transitions even if the number
of choices, $k(n)$,  increases unboundedly.

\section{Achlioptas Processes with Varying Choice}
\label{sec:1}

We consider Achlioptas processes run on a set of $n$ initially isolated nodes.  At each edge addition a number of candidate edges are first simultaneously examined, but only one of the edges is selected according to a pre-specified rule and added to the graph. Here we consider the case in which the number of candidate edges depends on system size, $k(n)$.

In Sec.~\ref{subsec:unbounded} we formally define unbounded size rules
and briefly discuss the known result that there exists an Achlioptas
process with an increasing number $k(n)$ of choices in which the
transition to large scale connectivity is discontinuous. In Sec.~\ref{subsec:bounded} we formally introduce bounded size rules and 
give a semi-rigorous argument that in bounded processes with infinite
choice, the phase transition remains continuous. In Sec.~\ref{subsec:degreebased} we introduce
a local rule which depends only on node degree that results in a discontinuous
phase transition.

\subsection{Unbounded Size Rules}
\label{subsec:unbounded}

An unbounded size rule, like the Product Rule, treats all components of distinct sizes uniquely. In contrast, a bounded size rule treats all components of size greater than $M$ as if they were of size $M$.
Unbounded size rules in which the number of choices $k(n)$ increases
with system size have the greatest potential for discontinuity. Before proceeding further,
we first make rigorous some definitions given in the introduction.

\begin{definition}
The component $C(x)$ containing a node $x$ is the set of nodes which
are either reachable from $x$ or which can reach $x$ by following
a simple path. A component $C$ is giant if $\frac{|C|}{n}$ does
not converge to $0$ as $n \rightarrow \infty$.
\end{definition}

\begin{definition}
The critical point of the percolation transition is $t_{c}=\inf\{t:\mbox{ w.h.p. a giant component exists at time \ensuremath{tn}}\}$.
The transition is discontinuous if there exist some functions $m(n)$
and $\Delta(n)=o(n)$ such that w.h.p no giant component exists when
$m(n)$ or fewer edges have been added to the graph and at least one
giant component exists when $m(n)+\Delta(n)$ or more edges have been
added to the graph. That is, the percolation transition is discontinuous
if the addition of a sublinear number of edges results in the emergence
of a giant component. 
\end{definition}

Most Achlioptas processes that have been studied fall under the following
general definition. At each timestep:
\begin{enumerate}
\item Two candidate edges are chosen uniformly at random from the set of
all possible edges. That is, each candidate edge would link two nodes
chosen uniformly at random. 
\item One of the two candidate edges is added according to some rule. 
\end{enumerate}

A more general definition of an Achlioptas process was given by Riordan
and Warnke~\cite{Riordan15072011}. Starting with a set of n isolated nodes, at each timestep
a set of $k$ vertices is chosen uniformly at random, where $k$ is a constant. Some subset
of edges among those vertices, possibly empty, is added to the graph.
However, if there exist two components of size $\epsilon n$ or greater,
the two components must be merged with probability at least $\epsilon^{k}$.
To meet this requirement it is sufficient to to demand that if all
vertices are chosen from exactly two components, then at least one
edge must be added. 

Utilizing this definition, Riordan and Warnke proved ~\cite{Riordan15072011} that if the number of choices
$k(n)$ increases unboundedly with system size and if the nodes are linked which are contained in the two smallest distinct components among those selected, then the percolation
transition is discontinuous as $t\rightarrow n$. The rule in which the two smallest distinct components are linked is referred to as  the ``SDC rule.'' Here the time $t$
is the same as the number of edges added to the graph. In particular,
they show that it takes at most $5\epsilon n$ edge additions
for $C_{1}$
to grow from less than $\epsilon n$ to greater than $(1-\epsilon)n$,
where $C_{i}$ denotes the size of the $i$th largest component. 

We will
refer to the time period between these two events as the critical window, 
parameterized by $\epsilon$. In Fig. 1, $C_{1}$
is plotted for the last 50 edges added during the process for two different systems sizes: $n=100,000$ and $n=1,000,000$.
This figure shows that it is extremely conservative to bound the size
of the critical window by $5\epsilon n$. In fact, it appears that
the size of the critical window is of constant rather than linear
size for fixed $\epsilon$.

\begin{figure}
\resizebox{0.5\textwidth}{!}{
 \includegraphics{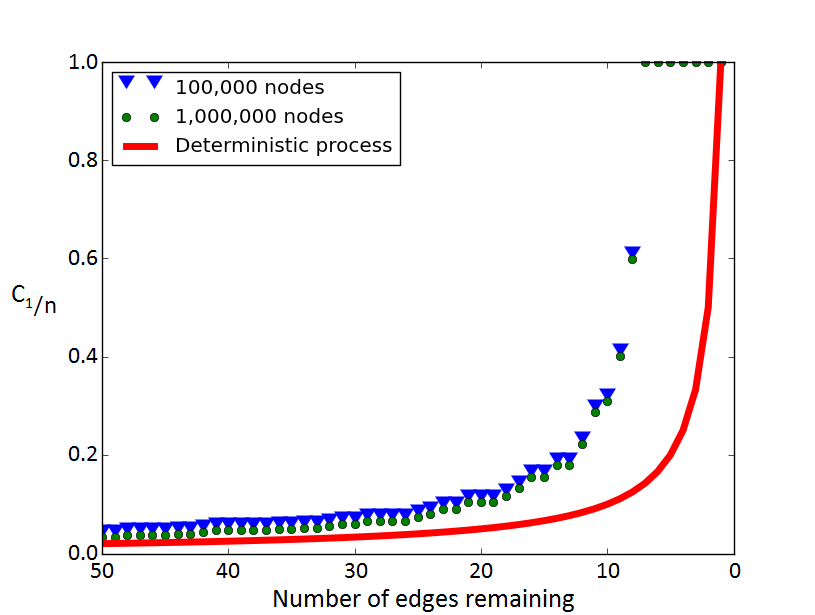}

}
\caption{
The solid curve denotes the size of the largest connected component in the deterministic process in which the two smallest components are always merged. The dotted curves are single simulated runs of the SDC rule with $k(n)=\log_2(n)$. The x axis denotes the final 50 edges added from $t = n - 50$ to $t = n$. Since this is a finite number of edges,  note that in each case the largest component grows from
$\frac{1}{10}n$ to $\frac{9}{10}n$ with the addition of fewer than
10 edges. We conjecture that the largest component grows from $\epsilon n$
to $(1-\epsilon)n$ with the addition of  $(\theta\frac{1}{\epsilon})$
edges, as is the case in the deterministic process denoted by the solid curve.
}

\label{fig:1}       % Give a unique label
\end{figure}

\subsection{Bounded Size Rules}
\label{subsec:bounded}

In the family of all possible rules for Achlioptas processes, bounded-size
rules treat all components of size greater than some constant $M$ as equivalent in size.
One of the simplest such rules is the Bohman-Frieze process~\cite{Bohman2001}. At each edge addition first two randomly selected candidate edges are examined.  If either edge links two isolated nodes it is added to the graph. (If both edges satisfy this condition, one of them is chosen at random and added.) If neither edge links two isolated nodes, a randomly selected edge is added instead. 
Since all non-isolated components are treated the same, this is a
bounded size rule with $M=1$. The degree-based rule which we will
define in section 2.3 and which is the focus of this paper can be
thought of as a slight variant of the Bohman-Frieze process.

It is conjectured~\cite{Spencer:2007:BCG:1390104.1390105,Janson_phasetransitions}
that any bounded size rule which chooses between two edges results
in the continuous emergence of a giant component in the same universality
class as the standard {Erd\H{o}s-R\'{e}nyi}~\cite{ER} phase transition.
Note that the transition must be continuous
 due to the result
of Riordan and Warnke given in section 2.1, since any such
process is an Achlioptas process with finite choice.  If instead we choose $k(n)$
nodes (or edges) where $k(n)\rightarrow\infty$, it becomes easier
to analyze the nature of the phase transition in a large class of
bounded size rules. In the paragraph below we show that for a particularly
interesting process the transition is continuous
as the number of choices approaches infinity, and moreover is in the
same universality class as the {Erd\H{o}s-R\'{e}nyi} random
graph process. Intuitively, the transition should remain continuous
if we limit ourselves to fewer choices.

Consider the rule defined as follows. At each timestep $k(n)$ nodes
are chosen uniformly at random. If at least two vertices are contained in a component of size at most $M$, an edge is added between the two vertices contained in the smallest component amongst those chosen. If all nodes chosen are in components of size $M$ or greater, then two nodes amongst those chosen are linked uniformly at random. 
In this case, there will be some time $t_{0}n$ at which almost every
component is between sizes $M$ and $2M$. To see that this is true,
note that whenever some portion of the nodes are contained in components
of size $M$ or smaller, two components of size $M$ or smaller will
be merged. However, once almost all of the nodes are in components of size
at least $M$, then edges will in fact be added uniformly at random. For simplicity,
take $M$ to be a power of $2$ so that almost all nodes are
in components of size exactly $M$. Suppose that we merge all nodes in components of size $M$ together and treat this collection as one isolated node. From this
onward, the process has become the standard {Erd\H{o}s-R\'{e}nyi}
random graph process. If at any point we halt the process and ``unmerge'' the nodes we will end up with components which are precisely $M$ times as large. Since $M$ is a constant, it follows that a giant component will emerge continuously
at time $t_{0}n+\frac{n}{2M}$. 
It is possible that there are bounded
size rules which result in a discontinuous phase transition given
a sufficiently large number of choices, but this is unlikely given
the above example. The single process we have chosen to rigorously
examine seems to be the ``worst case'' in that it minimizes the
size of the component created at each individual timestep.

\subsection{A Degree-based Rule with Varying Choice}
\label{subsec:degreebased}

\begin{figure}
\resizebox{0.5\textwidth}{!}{
 \includegraphics{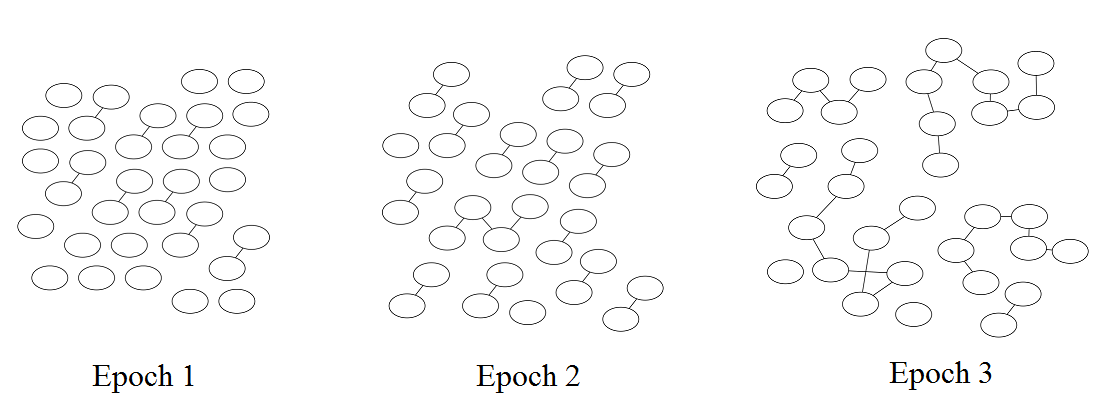}

}
\caption{Epochs of growth for the DRV model defined in Sec. 2.3, illustrating different phases of growth.  Epoch 1 is the timeframe $t\leq\frac{1}{4}n$, epoch 2 is $\frac{1}{4}n<t\leq\frac{1}{2}n$
and epoch 3 is $\frac{1}{2}n<t\leq\frac{3}{4}n$. As shown in the
figures above, in epoch 1 all nodes will be of degree 0 or 1, so that
almost all nodes are isolated or in components of size 2. In epoch
2 we start to see a very small number of degree 2 nodes which results
in a few 2-chains of degree 2 nodes terminated by degree 1 nodes.
In epoch 3 a fraction of the network will consist of 2-chains, formed
by linking two nodes of degree one. This continues into epoch 4 ($\frac{3}{4}n<t\leq n$),
with larger and larger 2-chains formed.}

\label{fig:2}      
\end{figure}

In this section we will first introduce a variant on the Bohman-Frieze model in which we choose edges from amongst $k(n)$ nodes, where $k(n)$ may vary with system size $n$. It is also stated in terms of node degree rather than component size. We then show this process has a continuous,  {Erd\H{o}s-R\'{e}nyii like phase transition. We will refer to this process as the Bohman Frieze process with {\it varying choice}.  Then, far more interesting, we show that by altering the process only to consider nodes of degree 1 in addition to isolated nodes of degree 0, we cause the giant component to emerge discontinuously. This further altered process will be referred to as the Degree Rule with Varying Choice, or DRV.

Consider the following generalization of the Bohman-Frieze process.
At each timestep $t$ choose a set $V_{t}$ of $k(n)$ vertices uniformly
at random, and let $V_{t,d}$ denote %be 
the set of vertices of degree $d$
in $V_{t}$. If $V_{t,0}$ contains at least two vertices then add
an edge between two vertices in $V_{t,0}$ sampled at random without
replacement. Otherwise, add an edge between two vertices
in $V_{t}$ sampled at random without replacement. Suppose $k(n)\rightarrow\infty$.
If $t=cn$ where $c<\frac{1}{2}$, then:

\begin{eqnarray}
\Pr(\lvert V_{t,0}\rvert\geq2) & = & 1-\Pr(\lvert V_{t,0}\rvert=0)-\Pr(\lvert V_{t,0}\rvert\geq1) \nonumber \\
 & = & 1-(1-\frac{D_{0}(t)}{n})^{k(n)}\nonumber \\
& & - k(n)(1-\frac{D_{0}(t)}{n})^{k(n)-1}\frac{D_{0}(t)}{n}\nonumber \\
 & = & 1-o(1).
\end{eqnarray}

It follows that almost every edge for $t<\frac{n}{2}$ is added between
two isolated nodes, so that $D_{0}(\frac{n}{2})=o(n)$. Therefore,
almost every edge added at times $t\geq\frac{n}{2}$ will be chosen
uniformly at random from the set of non-isolated nodes. If we restrict
our attention to the non-isolated nodes, the process evolves precisely
as the {Erd\H{o}s-R\'{e}nyi} process with the initial condition
$N_{2}(0)=\frac{n}{2}+o(n)$, where $N_k(t)$   is the number of components of size $k$ at time $t$.
It follows that a giant component emerges
continuously with the addition of $\frac{n}{4}$ edges, resulting
in a phase transition at time $t=\frac{3}{4}n$. This is removing
all nodes that were isolated at time $t=\frac{n}{2}$ from the graph,
and so it must be proven that if we add these nodes back in together with all
edges adjacent to them it does not change the nature of the phase
transition. However, note that all edges adjacent to these nodes must
have either linked a pair of nodes that were isolated at the time
of the addition or have been added between two non-isolated nodes
chosen uniformly at random. It follows that the effect on the size
of the giant component at any given time is less than if a sublinear
number of random edges were added to the graph. Therefore, the nature
and time of the transition have not changed.

Suppose that we further alter the process slightly. At
each timestep if $|V_{t,0}|<2$, then if $|V_{t,1}|\geq2$ we add
an edge between two  vertices in $V_{t,1}$ chosen at random. Otherwise we add
an edge between two random vertices in $V_{t}$. This process, which we refer to as the DRV process, at first behaves very similarly as illustrated in Fig. 2. During the first $\frac{n}{2}$ steps it is still almost always the case that two isolated nodes are linked. But afterwards the two processes are very different, with the emergence of the giant component delayed and occurring discontinuously in the DRV process. This will be proven in Sec. 3. 

A natural generalization is define $V_{t,i}$ for all integers $i>0$
and to prefer nodes of lower degree. In fact this process behaves
similarly to the degree-base process described in the previous
paragraph so long as $k(n)\rightarrow\infty$, and all of the results
in the following section can be applied to either model. In fact,
the two models are virtually identical for $k(n)\ge\log_{2}n$, in
the sense that a coupling argument reveals that there is vanishing
probability that they do not output the same graph. In a coupling
argument, we run both models using the same random ``coinflips''
at each step and calculate the probability that the two ever disagree
in the selection of a single edge. For more details, see~\cite{durrett2010probability}.
In the following section there is no need to distinguish between the
two models whenever $k(n)\geq\log_{2}n$.

However, in addition to linking nodes of degree 0 whenever possible,
we also link nodes of degree 1 if no nodes of degree 0 are chosen.
It is also similar to a global rule defined as follows. When $t<n(1-\frac{1}{k})$
merge two components chosen uniformly at random at each timestep.
When $t\geq n(1-\frac{1}{k(n)})$, link two nodes chosen uniformly
at random at each timestep. It is interesting that a rule depending
only on degree would behave similarly to a process which merges entire
components chosen uniformly at random. In the DRV process, the giant component emerges
discontinuously at time $n(1-\frac{1}{k(n)})$. Even though the transition may be considered
by some to be trivial because it occurs at $t=n-o(n)$. However, the
transition does not occur because we ``run out of components'' but
because the component size distribution becomes extremely heavy tailed,
resulting in a ``powder keg'' \cite{PhysRevLett.103.255701}. Regardless, we have greatly delayed the
onset of percolation using a simple degree-based rule.

\section{Results: Degree-based rule}
\label{sec:results}

The purpose of this section is to rigorously prove that the degree-based
rule defined in section 2.3 yields a discontinuous percolation transition
with a critical window of size $O(\frac{n}{k})$. In section 3.1,
we will put bounds on the resulting degree distribution and give a brief non-rigorous
argument to give intuition both for the discontinuity of the transition
and the size of the critical window. In sections 3.2 and 3.3, we will
complete the proof.

\subsection{Intuition for the Proof}

In this section, we will show that the degree-based process behaves
similarly to a hybrid of two simpler processes:

\begin{enumerate}

\item The standard {Erd\H{o}s-R\'{e}nyi} process, in which nodes are linked uniformly at random.
\item The process in which {\it components} are linked uniformly at random.

\end{enumerate}

Specifically, almost all steps in the process are described by one of the five events listed below (We explain the origin of the different events in the subsequent paragraphs.) Only a vanishing number of components
will have been involved in any event which is not equivalent to one
of these five. 
\begin{enumerate}
\item Two random isolated nodes are linked.
\item Two components are merged uniformly at random.
\item Two components are merged chosen proportional to component size minus
two.
\item Two components are merged, one of which is chosen uniformly at random
and the other chosen proportional to component size minus two(excluding components of size 1).

\end{enumerate}

In the DRV process, when $t \leq cn$ for any $c<\frac{1}{2}$ note that there are at least $(1-2c)n$ nodes of degree $0$. It follows with high probability an edge will be added between two nodes of degree $0$  at any individual timestep (see details in the next section) so that almost every edge to that point has been added between two isolated nodes. That is, event 1 dominates the process up until time $t=\frac{n}{2}$. Since almost all nodes are then in components of size 2, we may intuitively think of the process as ``starting over'' with the initial condition that all nodes are in components of size 2. It then follows for $t\geq \frac{n}{2}$ that almost every edge is added between two nodes of degree 1. If we assume that every single edge links two nodes of degree 1, the result is that every component has some number of nodes of degree 2 and exactly two nodes of degree 1. Since every component has precisely the same number of nodes of degree 1, it follows that whenever we link two nodes of degree one we are in fact merging two components uniformly at random. Hence event 2 dominates the process when $\frac{1}{2}n \leq t \leq cn$ for any $c<1$, since in this time period it is almost always the case that two random nodes of degree 1 are linked. 

It is well known \cite{Coalescense1999} that merging components uniformly at random does not result in the emergence of a giant component. Moreover, starting with the initial condition $N_2(\frac{n}{2}) = \frac{n}{2}$, the expected component distribution is $N_{2i}(t)= (1+\frac{t}{2})^{-2} (\frac{t}{2+t})^{i-1}n$ for $t\geq\frac{n}{2}$. As $t \rightarrow n$ the component distribution becomes a ``powder keg'' in which a portion of nodes is contained in large components. More precisely, there exist $cn$ nodes in components whose size diverges to $\infty$ as $n\rightarrow \infty$ for some constant $c$.    

As $t \rightarrow n$ the number of nodes of degree 1 dwindles, so we necessarily link some nodes of degree 2. If we link two nodes of degree 2, it is the same as merging two components chosen proportional to component size minus 2 if we assume that all components are still composed only of nodes of degrees 1 and 2. This is event 3. If we link a node of degree 1 to a node of degree 2, it is event 4. If $cn$ nodes are in components of size at least $f(n) = \omega(1)$, then $\theta(n/f(n))$ combined occurrences of events 3 and 4 are sufficient to result in the emergence of a giant component. 

The essence of the proof then, is to show that there are enough occurrences of event 2 to build up a powder keg without enough occurances of events 3 and 4 to defuse the powder keg before the giant component emerges. Moreover, the events above do not entirely describe the process. These technical details will be addressed in the following sections. However, by ignoring the technical details and considering only the four most likely events listed above, we can estimate that the critical point $t_c$ at which a giant component first emerges is located at $(1-\theta(\frac{1}{k(n)}))n$. As described in the previous paragraph, as the process proceeds the number of nodes of degree 1 dwindles and almost all nodes will be of degree 2, resulting in long "chains" of nodes of degree 2. It is only at time $(1-\theta(\frac{1}{k(n)}))n$ that we start to see non-vanishing numbers of higher degree nodes resulting from occurrences of events 3 and 4. Moreover, there will be a "powder keg" of components of size $\Omega(k)$ due to occurances of event 2. This results in the estimate $t_c = (1-\theta(\frac{1}{k(n)}))n$. See Fig. 6 and Fig. 7 for numerical evidence that this non-rigorous calculation is correct. This, in conjunction with Theorem 2, also shows that the size of the ``critical window" in which the largest component grows from size $\epsilon n$ to $(1-\epsilon)n$ is of size $\theta(\frac{n}{k})$.  

\subsection{Degree Distribution}

\begin{figure}
\resizebox{0.5\textwidth}{!}{
 \includegraphics{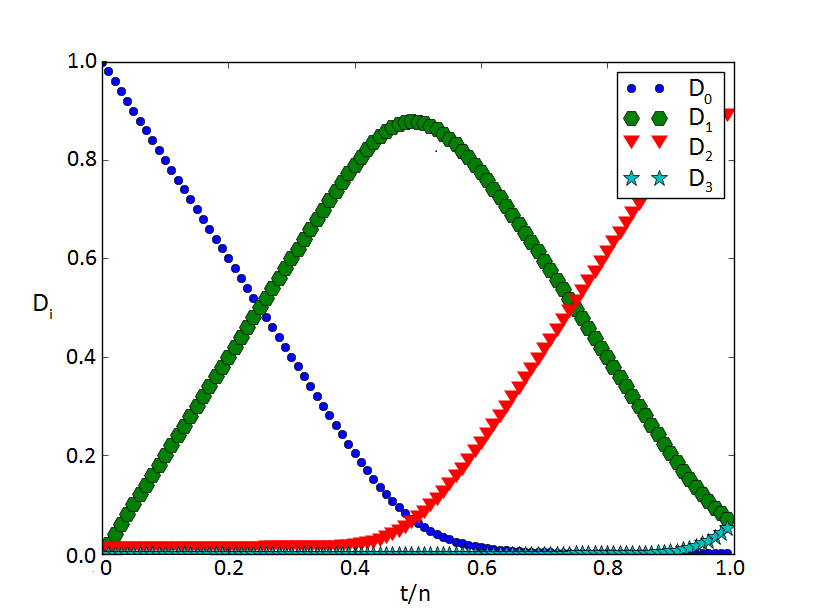}
}
\caption{Evolution of the degree distribution from $t=0$ to $t=n$ for the DRV process defined in Sec.~\ref{subsec:degreebased} with $k=\log_2 n$.}
\label{fig:3}      
\end{figure}

In this section we will approximate the degree distribution of the
process described in the previous section. Let $D_{i}(t)$ be the
number of nodes of degree $i$ at time $t$, and $p_{i}(t)=\frac{\langle D_{i}\rangle}{n}$.
Then:

\[
\langle D_{0}(t+1)\rangle-\langle D_{0}(t)\rangle=-2+2((1-p_{0}){}^{k})-(k(1-p_{0}){}^{k-1}p_{0})
\]

Given the initial condition $D_{0}(t)=n$, it easily follows that
$\langle D_{0}(t)\rangle=n-2t+O((\frac{2t}{n})^{k+1}\frac{n}{k})$
when $t<cn$ for any $c<\frac{1}{2}$. See the Appendix for a rigorous proof.
In order to see how many nodes of higher degrees there are, the following
relations will be useful. Note that (4) follows because the total
number of nodes in the graph is $n$, and (5) follows because the
sum of the degrees over all nodes is twice the number of edges in
the graph.

\begin{equation}
\sum_{i=0}^{\infty}\langle D_{i}(t)\rangle=n
\end{equation}

\begin{equation}
\sum_{i=1}^{\infty}i\langle D_{i}(t)\rangle=2t
\end{equation}

Subtracting (1) from (2) we see that $\sum_{i=2}^{\infty}(i-1)\langle D_{i}(t)\rangle=\langle D_{0}(t)\rangle+2t-n$.
Hence in particular we have that w.h.p. $\sum_{i=2}^{\infty}(i-1)\langle D_{i}(\frac{n}{2})\rangle=\langle D_{0}(\frac{n}{2})\rangle=O(\frac{n}{k})$.
It follows that $\langle D_{1}(\frac{n}{2})\rangle=n-O(\frac{n}{k})$,
so that at time $t=\frac{n}{2}$ almost all nodes are in components
of size $2$. Hence the process essentially ``starts over'' at time
$t=\frac{n}{2}$. In symbols, if $c<\frac{1}{2}$ then for $i\neq1$:

\begin{equation}
p_{1}(\frac{n}{2}+cn)=p_{0}(cn)+o(1)
\end{equation}

Therefore, an almost identical argument will show that for $t=\frac{n}{2}+cn$
where $c\leq\frac{1}{2}$, $\langle D_{1}(t)\rangle=n-2cn+O(\frac{n}{k})$.
Using relations (1) and (2) at time $t=n$ and subtracting twice the
value of (1) from (2), we obtain:

\begin{equation}
2\langle D_{0}(n)\rangle+\langle D_{1}(n)\rangle  =  \sum_{i=3}^{\infty}(i-2)\langle D_{i}(n)\rangle
\end{equation}

Since $\langle D_{0}(n)\rangle$ is insignificant, it follows that
$\langle D_{1}(n)\rangle\approx\sum_{i=3}^{\infty}(i-2)\langle D_{i}(n)\rangle$.
Moreover, if $k(n)\geq\log_{2}n$ then w.h.p. there will be no nodes
of degree 3 or higher, so that $\langle D_{1}(n)\rangle\approx\langle D_{3}(n)\rangle$.
This approximate symmetry can be clearly observed in Fig. 3. Table 1 summarizes the results of this section, assuming $k(n)\geq\log_{2}n$.

\begin{table}
\caption{Number of nodes of each degree at the end of each of the four epochs.}
\label{tab:1}       % Give a unique label
\begin{tabular}{|c|c|c|c|c|c|}
\hline\noalign{\smallskip}
t \textbackslash Degree & 0 & 1 & 2 & 3   \\
\noalign{\smallskip}\hline\noalign{\smallskip}
$\frac{n}{4}$ & $\frac{1}{2}n+o(n)$ & $\frac{1}{2}n+o(n)$ & 0 & 0  \\
$\frac{n}{2}$ & $O(\frac{n}{k})$ & $n-O(\frac{n}{k})$ & $O(\frac{n}{k})$ & 0 \\
$\frac{3}{4}n$ & $O(\frac{n}{k})$ & $\frac{1}{2}n+o(n)$ & $\frac{1}{2}n+o(n)$ & 0 \\
$n$ & $O(\frac{n}{k})$ & $O(\frac{n}{k})$ & $n-O(\frac{n}{k})$ & $O(\frac{n}{k})$ \\
\noalign{\smallskip}\hline
\end{tabular}
\end{table}

\subsection{Discontinuity of the Phase Transition}

\begin{figure}
\resizebox{0.5\textwidth}{!}{
 \includegraphics{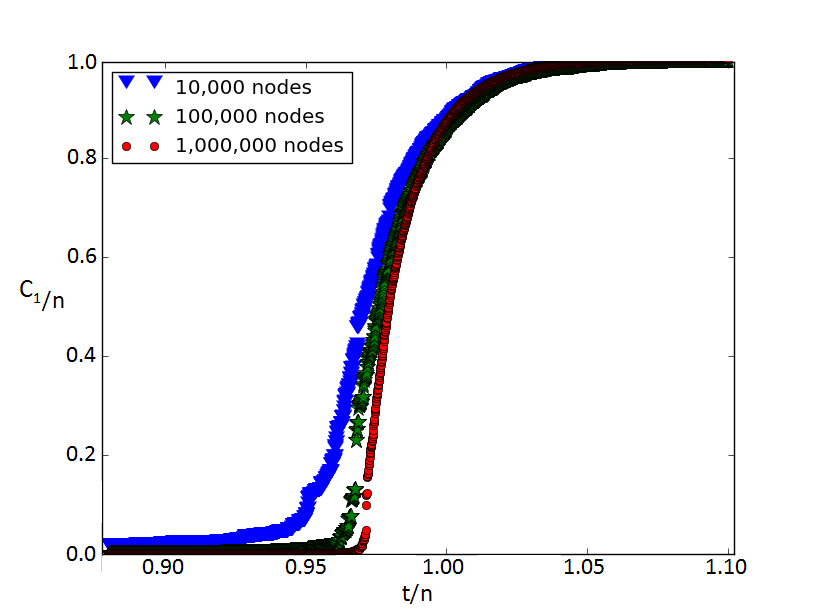}
}
\caption{The DRV process with $k(n)=\log_{2}n$. The transition shown here at $t=n-\theta(\frac{n}{k})$ will be shown to be discontinuous in the thermodynamic limit. Due to the lack of large, discrete jumps it may seem that the transition is continuous. Note that the size of the critical
window appears to shrink slightly as $n$ increases, consistent with discontinuity.}
\label{fig:4}      
\end{figure}

\begin{figure}
\resizebox{0.5\textwidth}{!}{
 \includegraphics{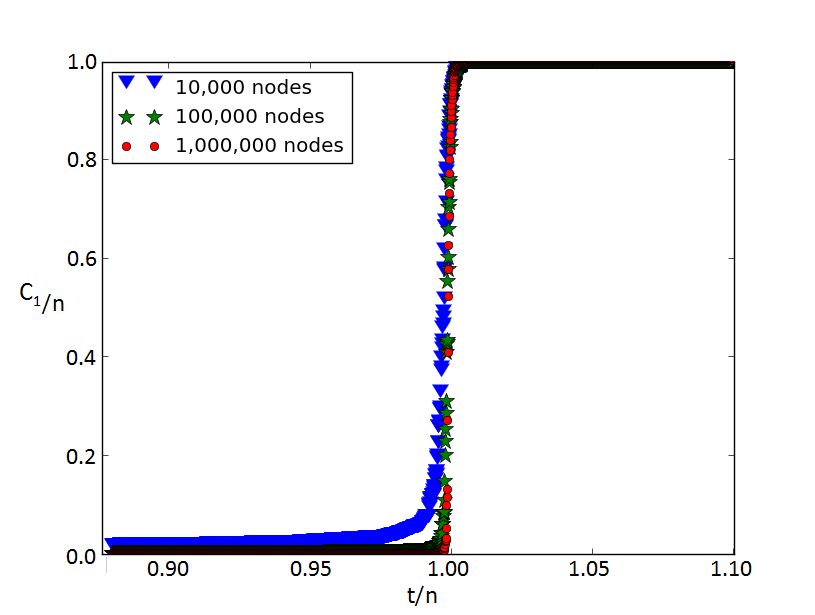}
}
\caption{ The DRV process with $k(n)=(\log_{2}n)^{2}$.A discontinuous phase transition at $t=n-\theta(\frac{n}{k})$ for
$k=(\log_{2}n)^{2}$. The transition seems much more obviously discontinuous,
but the size of the critical window has only shrunk by a factor of
$\log_{2}n$.}
\label{fig:5}      
\end{figure}

\begin{figure}
\resizebox{0.5\textwidth}{!}{
 \includegraphics{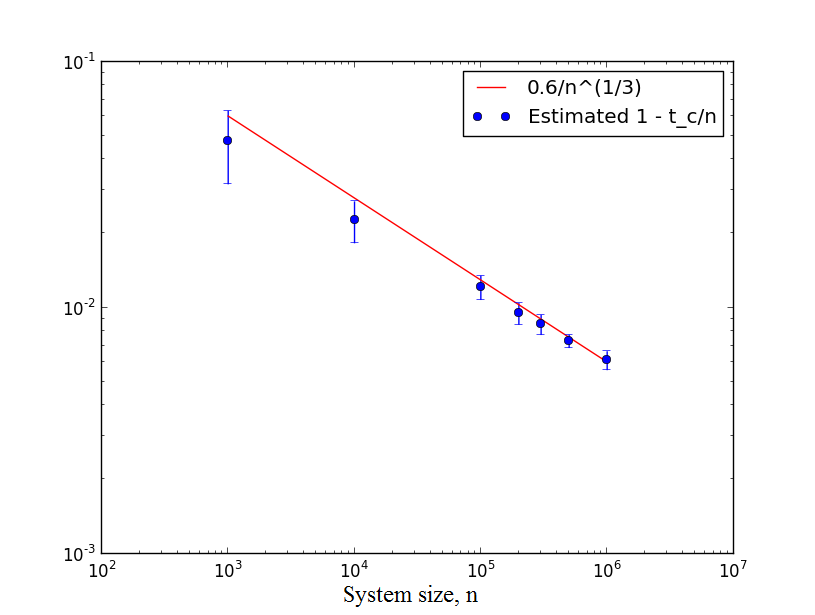}
}
\caption{Estimated critical points calculated via simulation of the DRV process with $k(n)  = n^{\frac{1}{3}}$, averaged over 30 trials, plotted against the curve $0.6*n^{-\frac{1}{3}}$. This shows that the critical points decay as a power law with exponent $\frac{1}{3}$, which is consistent with $t_c = (1 - \theta(\frac{1}{k}))n$. Smaller system sizes are included to show finite size effects and the increasing agreement with the curve for larger systems.} 
\label{fig:6}      
\end{figure}

\begin{figure}
\resizebox{0.5\textwidth}{!}{
 \includegraphics{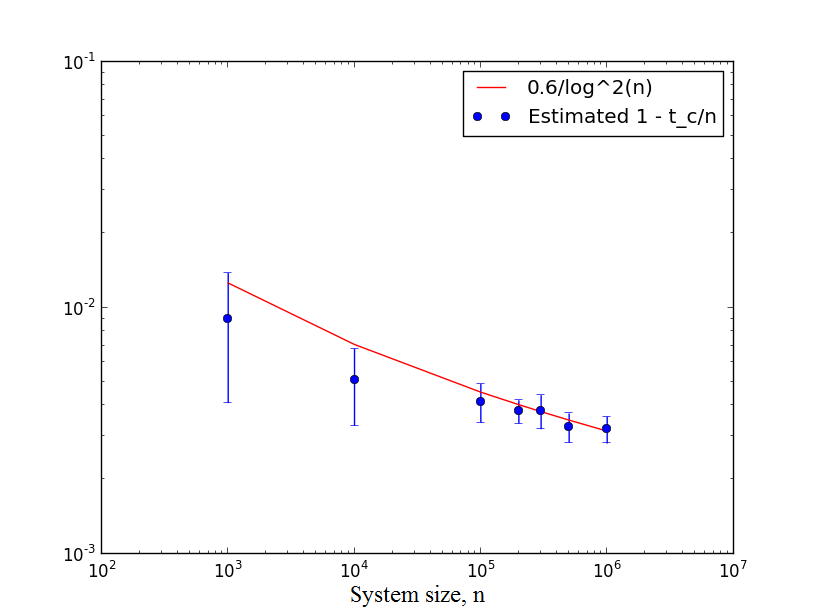}
}
\caption{Estimated critical points calculated via simulation of the DRV process with $k(n)  = \log^2(n)$, averaged over 10 trials, plotted against the curve $\frac{0.6}{\log^2(n)}$. The critical points no longer decay as a power law, so it is not true in general that $1-\frac{t_c}{n}$ decays as a power law, but the decay remains consistent with $t_c = (1 - \theta(\frac{1}{k}))n.$ Note that the constant 0.6 agrees with the constant in Fig. 6. } 
\label{fig:7}      
\end{figure}

In this section we will show that the giant component emerges discontinuously. Fig. 4 and Fig. 5 show that the emergence of a giant component is sudden and occurs near $t = n$, but cannot show whether these transitions are continuous or discontinuous in the thermodynamic limit. 
Recall that we define the critical window for each fixed $\epsilon$ as the interval of density in which $\epsilon n \leq C_1\leq (1-\epsilon) n$. The purpose of Theorem 1 is to put an upper bound on the rightmost point of the critical window, and in particular to show at at time $t=n$ a giant component has not yet emerged. In order to show that that the giant component emerges discontinuously it is then sufficient to show that a giant component has emerged at some time $t=n+g(n)$ where $g$ is some sublinear function of $n$. This is proven as Theorem 2. 

For convenience, we define $(i,j)$ edges as those edges added between
nodes of degree i and j respectively at the time the edge was added.
For example, the very first edge added is necessarily a $(0,0)$ edge since all nodes initially have degree $0$.
Additionally, edges will be colored. All $(1,1)$ links added at any
time $t\geq\frac{n}{2}$ and $(0,0)$ edges added at any time $t\leq\frac{n}{2}$
will be colored blue. All $(1,1)$ and $(0,1)$ links added before
time $t=\frac{n}{2}$ are colored green, and all other links are colored
red. The blue component $C^{B}(x)$ containing the node $x$ is the
set of nodes reachable from $x$ by following blue links together
with the node $x$ itself. Similarly, $N_{i}^{B}(t)$ is the number
of blue components of size $i$ at time $t$. Table 2 contains an explicit list of the edge color meanings.

%
% For tables use
\begin{table}
\caption{Edge colorings used in proofs.}
\label{tab:1}       % Give a unique label
% For LaTeX tables use
\begin{tabular}{|c|c|c|}
\hline\noalign{\smallskip}
$(i,j)$\textbackslash{}t & Before $t=\frac{n}{2}$ & After $t=\frac{n}{2}$  \\
\noalign{\smallskip}\hline\noalign{\smallskip}
$(0,0)$ & Blue & Red \\
$(0,1)$ & Green & Red \\
$(1,1)$ & Green & Blue \\
Other & Red & Red \\
\noalign{\smallskip}\hline
\end{tabular}
\end{table}

Intuitively, our degree-based process is very similar to one in which
we first merge all isolated nodes into components of size 2, merge
components uniformly at random up until time $t=(1-\frac{1}{k})n$
or so, and then start linking nodes uniformly at random as in the
typical {Erd\H{o}s-R\'{e}nyi} process. It is easy to show that
a process defined this way results in a discontinuous transition.
In order to rigorously prove that the same happens in the degree-based
process, we utilize colored edges in order to simplify certain complications.
The blue edges roughly correspond to the edges which agree precisely
with the first two steps outlined above: first almost all isolated
nodes are merged into components of size 2, and afterwards components
are linked uniformly at random. Green edges are early deviations in
which either isolated node are linked to non-isolated nodes or two
random non-isolated components are linked before almost all isolated
nodes have been linked into components of size 2. Finally, red edges
are non-blue edges added when almost no isolated nodes remain. We may think intuitively of red edges as being similar to edges added between two random nodes as in the {Erd\H{o}s-R\'{e}nyi} process. 
\begin{theorem}
If $k(n)=\Omega(\log n$) in the DRV process, then no giant component exists at time
$t=dn$ for any $d<1$.\end{theorem}
\begin{proof}
We will first show that distribution of the blue component sizes has
an exponential tail at any time $t=dn$ for any $d<1$ with the largest
component logarithmic in size. We will then show that the number of nodes in blue components with adjacent red or green edges is sublinear, so that all such nodes can be discarded without affecting the existence of a (non-blue) giant component. After discarding these nodes, the definition of blue component becomes identical to the traditional graph theoretic definition of a component, and hence no giant component exists.  

By lemma 1, the number of red and green edges added at time $dn$ is sublinear for any $d<1$, so almost all blue components will have no red or green edges. Moreover, if we discard all blue components with red or green edges the proof is simple, because the process will proceed as follows:

\begin{enumerate}
\item Up until time $t=\frac{n}{2}$
  we link nodes in blue components of size 1 at each step, so that almost all nodes are in blue components of size 2. 
\item From time $t=\frac{n}{2}$ until time $t=cn$
  we always link two random degree 1 nodes. Note that all remaining blue components with more than one node have exactly two nodes of degree 1, and components of size 1 are not involved in any mergings, so this is exactly the same as merging two random blue components of size greater than 1. 
\end{enumerate}

The proof then follows by considering the well-understood process in which at each step two components are linked uniformly at random, with the minor variation that we initialize with all components of size 2. This merely speeds up the process by a factor of 2 and multiplies all component sizes by 2. Hence $\langle N_{2i}^{B}(cn)\rangle=(c(1-c)^{2}c^{2i}n+o(n)$
 , with the $o(n)$
  accounting for the the sublinear number of blue components which have adjacent red or green edges. 

Although red and green edges do not have any effect on the blue component structure, they change node degrees and hence can alter the placement of additional blue edges. For example, if a red $(0,0)$ edge is added after time $\frac{n}{2}$
 , both of the linked nodes remain blue components of size 1, but have degree 1 and hence may afterwards be chosen as an endpoint of a $(1,1)$ blue edge. A green $(1,1)$ edge added before time $\frac{n}{2}$ will reduce the number of degree 1 nodes in two different blue components by one. However, components with green and red edges will, if anything, have fewer nodes of degree the number of degree one nodes in a blue component may only be increased by the addition of a red or green edge if the component is of size one, and in all other cases number of degree 1 nodes in a blue component is either decreased or unaffected. In the case where the number of degree 1 nodes is decreased, further mergings become less likely, and it follows that the tail of the component size distribution of blue components with adjacent red or green edges is lighter than that of blue components without adjacent red or green edges. Moreover, almost all nodes are contained in blue components without red or green adjacent edges, because the total number of such edges is sublinear and the average blue component size is finite.  Therefore, no giant blue component exists and the distribution on blue component size decays exponentially at any time $t = dn$ for $d < 1$. 

Any non $(0,0)$ link that is added up to time $t=\frac{n}{2}$
is colored green. It follows from lemma 1 that at time $t=\frac{n}{2}$ the number
of nodes of degree 1 in components with at least one green edge is
$O(\frac{n}{k})$. The number of nodes of degree 1 in this collection
cannot increase through the addition of red or blue edges, and no more green
edges will be added. It follows w.h.p. that throughout the process the number of nodes of degree
1 in components with at least one green edge is $O(\frac{n}{k})$.  Therefore, the number of times a blue component without adjacent green edges is merged to a blue component with green edges is w.h.p. sublinear. Moreover, merging two blue components with adjacent green edges does not change the number of nodes contained in such components. Finally, since the distribution on blue components is exponential and hence the expected size of a blue component is finite, it follows that the total number of nodes contained in blue components with adjacent green edges is w.h.p. sublinear. If the number of nodes in blue components with adjacent red links is also sublinear, we can simply discard all nodes in components with red or green links. This does not affect the link structure of the remaining nodes since there are no red or green links among them, and in particular does not affect the existence of a giant component. 

Since the number of red edges added by time $dn$ is sublinear w.h.p. by lemma 1, it is sufficient to show that the expected component size of a blue component containing one or more red edges at time $dn$ is finite.  In this case, the number of nodes in blue components with one or more adjacent red edges is w.h.p. sublinear, so that even merging all such components together does not result in a giant component. 

The probability that a red (2,2) edge connects blue components of size $i$ and $j$ at time $t$ is bounded by \\ $\frac{iN_i^B(t) jN_j^B(t)}{D_2(t)^2}$ since a blue component of size i has at most $i$ degree $2$ nodes. The expected size of a pair of blue components linked by a red (2,2) edge at time $t$ is bounded by $\sum_{i,j=1}^n (i+j)\frac{iN_i^B(t)jN_j^B(t)}{D_2(t)^2}$. Since the number of degree $2$ nodes is w.h.p. linear at time $dn$ and the distribution on blue component size decays exponentially, this sum is finite. It follows the expected size of a blue component with adjacent red edges is finite as well. We may similarly bound the expected blue component size with adjacent red (0,0),(0,1),(0,2), and (1,2) edges. With high probability there are no red (i,j) edges for i or j greater than 2 at time $dn$, which follows from the assumption that $k(n)\geq \log_2(n)$.  

Since blue components without any adjacent red or green edges are components in the traditional sense, and the number of nodes in components with red or green edges is sublinear, it follows that w.h.p. no giant component exists at time $t=dn$. 
\end{proof}
In order to show that a giant component emerges discontinuously, it remains to show only that a giant component emerges at some time
$t=n+o(n)$, which follows from a straightforward proof by contradiction. 
\begin{theorem}
If $\frac{n}{k(n)}=\omega(\log n)$ in the DRV process, then for any $\epsilon>0$ there
exists some constant $B$ such that $C_{1}(n+B\frac{n}{k}))\geq(1-\epsilon)n$.\end{theorem}
\begin{proof}
Given a set $C$ of $l$ components, if we merge two components in
$C$ and then successively merge two components $l$ times, the result
will be a totally connected graph. It follows that if there are a
sublinear number of components at step $t$ and a positive probability
at each timestep of merging two distinct components, then in a sublinear
number of timesteps the graph will be fully connected.

Suppose that at step $t=n+BN(n),$ $|C_{1}|<(1-\epsilon)n$ for any
fixed $B$, where $N(n)$ is the number of distinct components at
time $n$. It follows that the probability of merging two distinct
components is at least $\epsilon$ from any time $n\leq x<n+BN(n)$.
Let $M$ be the total number of mergings when $n\leq x<n+BN(n)$.
Then $M$ stochastically dominates the Binomial random variable $M'\sim Bin(\Delta,\epsilon)$.
Take $\Delta=BN(n)$. Then a Chernoff bound gives:

\begin{eqnarray}
\Pr(M'<BN(n)\epsilon-t)\leq e^{-\frac{t^{2}}{2(BN(n)\epsilon+t)}} \nonumber \\
\end{eqnarray}

It follows that for $B>\frac{1}{\epsilon}$, w.h.p. the graph is fully
connected so long as $N(n)=\omega(1)$. But if $N(n)=O(1)$ then the
result is trivial, so we may assume that $N(n)=\omega(1)$. Therefore,
we have a contradiction, and hence there exists some $B$ such that
$C_{1}(n+BN(n)))\geq(1-\epsilon)n$. All that is left is to show that
$N(n)=O(\frac{n}{k})$, and this is proven as lemma 2 in the appendix. 
\end{proof}

\section{Discussion}

In this paper we defined an Achlioptas process whose rule depends only on degree and showed that it exhibits a discontinuous transition if the number
of choices increases sufficiently with system size. However, we also noted that the transition may appear to be continuous on finite systems despite a rigorous proof that it is not. This is in great contrast to standard Achlioptas processes with a fixed number of choices, which may appear discontinuous despite being continuous in the thermodynamic limit. Regardless,
the percolation transition is greatly delayed and no global properties
of the network are needed in order to make the choice of which edge to add, requiring only the degrees of the nodes chosen.

Even if the number of choices is allowed to increase with system size,
it does not seem that a bounded size rule may result in a discontinuous
phase transition. In real networks, agents which create links will
often choose amongst many possible targets. In such a competitive
environment, it is unlikely that global information such as component
size is available. However, local information such as node degree
or some approximation of node degree might be available. Additionally, even if the number
of competing choices does not grow with system size, the analysis given here may still hold for very large systems so long as the number of choices is reasonably large. Consider, for instance, the case where the number of choices is $20$ and the number of nodes is less than $2^{20}$. In this case, the number of choices is at least $\log _2 (n)$, and so all of the results in section 3 hold.  

\vspace{0.1in}
\noindent
{\bf Acknowledgements:} We gratefully acknowledge support from the US Army Research Laboratory and the US Army Research Office under MURI award W911NF-13-1-0340, and Cooperative Agreement W911NF-09-2-0053, the Defense Threat Reduction Agency Basic Research Grant No. HDTRA1-10-1-0088 and NSF grant number ICES-1216048.

\section{Appendix}
%\begin{proposition}
%For any $c\leq \frac{1}{2}$, $\langle D_{0}(cn)\rangle=n-2cn+O((2c)^{k}\frac{n}{k})$\end{proposition}
\begin{proposition}
$\langle D_{0}(\frac{n}{2})\rangle=O(\frac{n}{k})$
\end{proposition}
\begin{proof}
Define $x(t):=\frac{\langle D_{0}(t)\rangle}{n}$, so that we have
the difference equation:

\begin{eqnarray}
(x(t+1)-x(t)) &=& \frac{1}{n}(-2+2((1-x(t)){}^{k}) \nonumber\\
{}& & {}+ (k(1-x(t)){}^{k-1}x(t))) \nonumber
\end{eqnarray}

Note that $x(t)\geq1-\frac{2t}{n}$,
which follows since at time $t$ only $t$ edges have been added to
the graph. Moreover, the functions  $2((1-x(t)){}^{k})$ and $(k(1-x(t)){}^{k-1}x(t)$ are monotone decreasing, so we have the bound:

\begin{eqnarray}
x(t+1)-x(t) & \leq & -\frac{2}{n}+\frac{2}{n}((\frac{2t}{n}){}^{k}) + k(\frac{2t}{n}){}^{k-1}(1-\frac{2t}{n}) \nonumber \\ 
& = & -\frac{2}{n}+\frac{2}{n}(\frac{2t}{n}){}^{k} + k((\frac{2t}{n})^{k-1} - (\frac{2t}{n})^{k}) \nonumber \\
& = & -\frac{2}{n}+(\frac{2}{n})^k(\frac{2t^{k}}{n} +k(\frac{nt^{k-1}}{2} - t^{k})) \nonumber
\end{eqnarray}

Since we have eliminated x from the right side of this bound, we can attempt to sum it directly and apply the equality  $x(t+1)=x(0)+\sum_{i=0}^{t}x(i+1)-x(i)$.

\begin{eqnarray}
x(t+1) & = & 1 + \sum_{i=0}^{t}x(i+1)-x(i)] \nonumber \\
& \leq  & 1 + \sum_{i=0}^{t}-\frac{2}{n}+(\frac{2}{n})^k(\frac{2i^{k}}{n} +k(\frac{ni^{k-1}}{2} - i^{k})) \nonumber \\
& \approx & 1 - \frac{2t}{n}+(\frac{2}{n})^k(\frac{2t^{k+1}}{n(k+1)} + k(\frac{nt^{k}}{2k} - \frac{t^{k+1}}{k+1})) \nonumber
\end{eqnarray}

The final line is an integral approximation. The result then follows by multiplying $x(t)$ by
$n$ and setting $t = \frac{1}{2}n$.
\end{proof}

\begin{proposition}
$\langle D_{1}(n)\rangle=O(\frac{n}{k})$
\end{proposition}
\begin{proof}
By Proposition 1,  $\langle D_0(\frac{n}{2}) \rangle = O(\frac{n}{k})$, and it follows easily from (5) that $\langle D_1(\frac{n}{2}) \rangle = n - O(\frac{n}{k})$. The analysis of $\langle D_1 \rangle$ from $t=\frac{n}{2}$ to $t = n$ is then nearly identical to the analysis of $\langle D_0 \rangle $ from $t=0$ to $t=\frac{n}{2}$ in Proposition 1, and is omitted.

\end{proof}

\begin{lemma}
The expected number of red and green links is $O(\frac{n}{k})$ at any time
$t\leq n$.
\end{lemma}
\begin{proof}

Recall that a red link is defined as a $(0,1)$ or $(0,0)$ link added after time $t = n/2$ or an $(i,j)$ link where either $i$ or $j$ is not $0$ or $1$ added at any time. By proposition 1 the number of nodes of degree $0$ at time $n/2$ is $O(n/k)$, and it follows trivially that the number of $(0,0)$ and $(0,1)$ red links is $O(n/k)$ at any time throughout the process, and similarly for green links. All that remains is to show that there are $O(n/k)$ red links involving higher degree nodes. However, note that adding an $(i,j)$ link where either $i$ or $j$ is not $0$ or $1$ would create at least one node of degree $3$ or higher. But (5) together with Proposition 2 shows that there are $O(n/k)$ nodes of degree $3$ or higher at time $t = n$. The result follows.  

\end{proof}

%\begin{lemma}
%If $k(n)=\Omega(\log(n))$, then w.h.p. there exist $c>0$ so that
%the first $cn$ links connect two isolated nodes.\end{lemma}
%\begin{proof}
%The probability that two isolated nodes are connected is $1-k\frac{D_{0}}{n}(1-\frac{D_{0}}{n})^{k-1}-(1-\frac{D_{0}}{n})^{k}\geq1-k(\frac{2t}{n})^{k-1}-(\frac{2t}{n})^{k}$. Applying a union bound, it suffices to show that the following quantity converges to $0$.

%Hence the probability that anything else happens is at most $k(\frac{2t}{n})^{k-1}+(\frac{2t}{n})^{k}$

%\[\sum_{t=1}^{cn}k(\frac{2t}{n})^{k-1}+(\frac{2t}{n})^{k}=k(\frac{2}{n})^{k-1}\sum_{t=1}^{cn}t^{k-1}+(\frac{2}{n})^{k}\sum_{t=1}^{cn}t^{k}\]

%Since $k(n)=\Omega(\log n)$ it follows that $k(n)\geq\log_{b}n$
%for some base $b>0$. If we take $c<\frac{b}{2}$, then an integral approximation shows that the above
%quantity converges to $0$.
%\end{proof}

\begin{lemma}
If $\frac{n}{k(n)}=\omega(\log n)$, then w.h.p. $N(n)=O(\frac{n}{k}$)\end{lemma}
\begin{proof}
It is sufficient to show that the number of blue components at time
n is $O(\frac{n}{k})$, since the number of components at time n, $N(n)$, will be smaller.
The number of blue components is determined by the number of blue
edges added compared to the number of edges added internally to an
existing blue component. The number of internal blue edges added is
distributed as $\sum_{i=1}^{(1-\theta(\frac{1}{k}))n}1_{p(t(i))}$
where $p(t_{i})=\frac{1}{D_{1}(t_{i})}$, $t(i)$ is the time at which
the $i$th (1,1) edge is added, since (0,0) edges can never be internal
to a single blue component. Note that $\sum_{i=1}^{(1-\theta(\frac{1}{k}))n}1_{p(t(i))}$
is stochastically bounded by $X=\sum_{i=1}^{(1-\theta(\frac{1}{k}))n}1_{p_{i}}$
where $p_{i}=\frac{1}{n-i}$. Hence $\langle X\rangle=\sum_{i=(1-\theta(\frac{1}{k}))n}^{n}\frac{1}{i}=O(\log n)$.
The result then follows from a Chernoff bound on $X$, which
shows that:

\begin{eqnarray}
\Pr(X>\langle X\rangle+t)=\Pr(X>O(\log n)+t)\leq e^{-\frac{t^{2}}{2(E(X)\epsilon+t)}} \nonumber
\end{eqnarray}

\end{proof}

\bibliographystyle{unsrt}
\bibliography{infchoicebib}

\end{document}